\documentclass[sigconf]{acmart}

\usepackage{booktabs} 

\usepackage{booktabs} 
\usepackage[utf8]{inputenc}
\usepackage[T1]{fontenc}
\usepackage{microtype}
\usepackage[ruled]{algorithm2e} 
\usepackage{flushend}

\usepackage{color}
\usepackage{enumitem}
\usepackage{soul}
\usepackage{graphicx,amsthm,amsmath,amssymb,amsfonts}
\usepackage{booktabs}
\newtheorem{claim}[theorem]{Claim}
\newtheorem{remark}[theorem]{Remark}

\newcommand{\E}{\mathbb{E}}
\newcommand{\cA}{\mathcal{A}}
\newcommand{\N}{\mathbb{N}}
\renewcommand{\P}{\mathbb{P}}
\newcommand{\R}{\mathbb{R}}
\newcommand{\Z}{\mathbb{Z}}

\newcommand{\rank}{\text{rank}}
\newcommand{\hH}{\mathbf{H}}
\newcommand{\A}{{\mathbf{A}}}
\newcommand{\eps}{\epsilon}

\setcopyright{rightsretained}

\setcopyright{acmcopyright} 
\acmDOI{10.1145/3055399.3055402}
\acmISBN{978-1-4503-4528-6/17/06}
\acmConference[STOC'17]{49th Annual ACM SIGACT Symposium on the Theory of Computing}{June 2017}{Montreal, Canada} 
\acmYear{2017}
\copyrightyear{2017}
\acmPrice{15.00}

\begin{document}
\title{Local Max-Cut in Smoothed Polynomial Time}

\author{Omer Angel}
\affiliation{%
  \institution{University of British Columbia}
  \department{Department of Mathematics}
  \country{Canada}
}
\email{angel@math.ubc.ca}

\author{S\'ebastien Bubeck}
\authornote{The corresponding author}
\affiliation{%
  \institution{Microsoft Research;}
  \country{USA}
}
\email{sebubeck@microsoft.com}

\author{Yuval Peres}
\affiliation{%
  \institution{Microsoft Research;}
  \country{USA}
}
\email{peres@microsoft.com}

\author{Fan Wei}
\affiliation{%
  \institution{Stanford University}
  \department{Department of Mathematics}
  \country{USA}
}
\email{fanwei@stanford.edu}


\begin{abstract}
In 1988, Johnson, Papadimitriou and Yannakakis wrote that ``Practically all the empirical evidence would lead us to conclude that finding locally optimal solutions is much easier than solving NP-hard problems". Since then the empirical evidence has continued to amass, but formal proofs of this phenomenon have remained elusive. A canonical (and indeed complete) example is the local max-cut problem, for which no polynomial time method is known. In a breakthrough paper, Etscheid and R{\"o}glin proved that the {\em smoothed} complexity of local max-cut is quasi-polynomial, i.e., if arbitrary bounded weights are randomly perturbed, a  local maximum can be found in $\phi n^{O(\log n)}$ steps where $\phi$ is an upper bound on the random edge weight density. In this paper we prove smoothed polynomial complexity for local max-cut, thus confirming that finding local optima for max-cut is much easier than solving it.
\end{abstract}

%
%




\maketitle

\section{Introduction}
Let $G = (V,E)$ be a connected graph with $n$ vertices and $w: E\rightarrow
[-1,1]$ be an edge weight function. The local max-cut problem asks to find
a partition of the vertices $\sigma: V\rightarrow \{-1,1\}$ whose total cut
weight
\begin{equation} \label{eq:localmaxcut}
  \frac12 \sum_{uv \in E} w(uv) \big(1-\sigma(u)\sigma(v)\big) ,
\end{equation}
is locally maximal, in the sense that one cannot increase the cut weight
by changing the value of $\sigma$ at a single vertex (recall that finding the
global maximum of \eqref{eq:localmaxcut} is NP-hard). This problem comes up
naturally in a variety of contexts.  For example \cite{SY} showed that
local max-cut is complete for the complexity class \emph{Polynomial-Time
  Local Search} (PLS).  It also appears in the {\em party affiliation
  game}, \cite{FPT}: this is an $n$-player game where   each player 
  $v\in V$ selects an action $\sigma(v) \in \{-1,1\}$ and the resulting payoff for player $v$ is
$\mathrm{sign} \left( \sum_{u v \in E} w(uv) \big(1-\sigma(u)\sigma(v)\big) \right)$.
It is easy to see that a local maximum of
\eqref{eq:localmaxcut} exactly corresponds to a Nash equilibrium for the
party affiliation game. Yet another appearance of this problem is in the context of Hopfield networks, \cite{Hop}: 
this is a collection of neurons with weighted connections between them, where each neuron is in one of two states (either firing or not firing) and with state
update at random times by thresholding the sum of incoming weights from firing neurons. It is again easy to see that such dynamics make the state configuration
converge (for undirected weights) to a local maximum of \eqref{eq:localmaxcut} (with $\sigma(u)$ representing the state of neuron $u$ and $w(uv)$ the weight of the connection
between neurons $u$ and $v$).

There is a natural algorithm to find a local maximum of
\eqref{eq:localmaxcut}, sometimes referred to as the \emph{FLIP algorithm}:
Start from some initial partition $\sigma$, and until reaching a local maximum, repeatedly
find a vertex for which flipping the sign of $\sigma$ would increase the cut
weight - and carry out this flip.  (To be precise, this is a family of
algorithms corresponding to different ways of selecting the improving
change when there are multiple possibilities.)  This algorithm also
corresponds to a natural dynamics for the party affiliation game, and a specific implementation 
(random selection of an improving vertex) exactly corresponds to the asynchronous Hopfield network dynamics described above.
However, it is easy to see that there exists weight functions such
that FLIP takes an exponential number of steps before reaching a local
maximum. As noted in \cite{JPY} (who introduced the PLS
class), this seems at odd with empirical evidence suggesting that algorithms
such as FLIP usually reach a local maximum in a reasonable time. This
conflicting situation naturally motivates the study of the {\em smoothed
  complexity} of local max-cut: is it true that after adding a small amount of
noise to the edge weights, the FLIP
algorithm terminates in polynomial time  with high probability? 
In this paper we
answer this question affirmatively, provided that a small amount of noise
is added to all vertex pairs (i.e., even to non-edges);  in other words, we assume that $G$ is a complete graph.
We note that a similar subtlety arises in the smoothed analysis of the
simplex algorithm by \cite{ST} where noise is added to every entry of the
constraint matrix (in particular, the null entries are also smoothed).

We now introduce the problem formally, discuss existing results, and
state our main contributions.  Let $X = (X_e)_{e \in E} \in [-1,1]^E$ be a
random vector with independent entries.
One should think of
  $X_e$ as the original edge weight $w(e)$ plus some independent small
  noise. We assume that $X_e$ has a density $f_e$ with respect to the
Lebesgue measure, and we denote $\phi = \max_{e \in E} \|f_e\|_{\infty}$.
In this paper the phrase \textbf{with high probability} means with
probability at least $1-o_n(1)$ with respect to $X$. We consider the space
of spin configurations $\{-1,1\}^V$, and for a spin configuration $\sigma
\in \{-1,1\}^V$ we denote by $\sigma(v)$ the value of $\sigma$ at vertex
$v$. We are interested in the random map $\hH : \{-1,1\}^V \rightarrow \R$
(usually called the {\em Hamiltonian}) defined by:
\begin{equation} \label{eq:Hdef}
\hH(\sigma) =-\frac{1}{2} \sum_{uv \in E} X_{uv} \sigma(u) \sigma(v).
\end{equation}
Our objective is to find a local maximum of $\hH$
with respect to the \textbf{Hamming distance} $d(\sigma,\sigma') = \#\{v :
\sigma(v)\neq \sigma'(v)\}$.  Equivalently, we are looking for a locally
optimal cut in the weighted graph $(G,X)$ (since \eqref{eq:localmaxcut} and \eqref{eq:Hdef} differ by the half of the total weight of all edges).

We say that $\sigma'$ is an {\bf improving move} from $\sigma$ if $d(\sigma',
\sigma)=1$ and $\hH(\sigma') > \hH(\sigma)$. We will sometimes refer to a
sequence of improving moves as an {\em improving sequence}.
The FLIP algorithm iteratively
performs improving moves until reaching a configuration with no improving
move. An {\em implementation} of FLIP specifies how to choose the initial
configuration and how to choose among the improving moves available at each step.
 \cite{ER} show that for any graph with smoothed weights, with high
probability, any implementation of FLIP will terminate in at most $n^{C
  \log(n)}$ steps, for some universal constant $C>0$.

Our main result is  that FLIP terminates in a {\em polynomial\/} number of
steps for the complete graph.  Since our results are asymptotic in $n$, in
the rest of the paper we assume $n \geq n_0$ for some universal constant
$n_0$.

\begin{theorem}\label{T:main_Kn}
  Let $G$ be the complete graph on $n$ vertices, and assume the edge
  weights $X = (X_e)_{e\in E}$ are independent random variables with
  $|X|\leq 1$ and density bounded above by $\phi$.  For any $\eta>0$, with
  high probability any implementation of FLIP terminates in at most
  $O(\phi^5 n^{15+\eta})$ steps, with implicit constant depending only on
  $\eta$.
\end{theorem}

\begin{corollary}\label{C:main_Kn}
  Under the assumptions of Theorem~\ref{T:main_Kn}, the expected number
  of steps of any implementation of FLIP is $O(n^{15})$, with implicit
  constant depending only on $\phi$.
\end{corollary}

Note that \emph{any implementation} is a very broad category.  It includes
an implementation where an adversary with unbounded computational power
chooses each improving step.  Theorem \ref{T:main_Kn} implies that even in
this case the number of steps is polynomial with high probability.

\begin{remark}
  The edge weights are assumed to be bounded only for  simplicity.
  Our methods can be used to give the same bound as long as the edge
  weights have finite variance, and a polynomial bound as long as the
  weights have a polynomial tail. Indeed, if $\P(|X|>t) \ll t^{-\delta}$,
  then with high probability all edge weights are at most $n^{3/\delta}$;
  rescaling the edge weights by $n^{3/\delta}$ increases $\phi$ by a
  corresponding factor, giving a bound of order $n^{15+15/\delta +\eta}$.
\end{remark}

\begin{remark}
In the the classical Sherrington-Kirkpatrick model \cite{SK}, a mean field model for a spin glass, the Hamiltonian is exactly a scaled version of the random map defined in (\ref{eq:Hdef}) and when $X_{ij}$ are i.i.d. Gaussian random variables for all pairs $i,j$. Therefore our Theorem \ref{T:main_Kn} implies that in the in the  Sherrington-Kirkpatrick model, the maximal length of a monotone path (along which the energy is decreasing) in the random energy landscape is $O(n^{15+\eta})$.
\end{remark}

Theorem \ref{T:main_Kn} can be equivalently stated as follows.

\begin{theorem}\label{T:main_Kn2}
  Let $G$ be the complete graph. Assume the edge
  weights $X = (X_e)_{e\in E}$ are independent random variables with
  $|X|\leq 1$ and density bounded above by $\phi$.  The probability that there is an improving
  sequence of length $\Omega(\phi^5 n^{15+\eta})$ is $o(1)$.
\end{theorem}

We say that a sequence $L$ is \textbf{$\epsilon$-slowly improving} from an
initial state $\sigma_0$ if each step of $L$ increases $\hH$ by at most
$\eps$ (and more than $0$).
Our main task will be to prove the following proposition:

\begin{proposition} \label{P:POLY}
  Fix $\eta>0$ and let $\eps = n^{-(12+\eta)} \phi^5$.  Then with high
  probability, there is no $\eps$-slowly improving sequence of length $2n$
  from any $\sigma_0$.
\end{proposition}

Proposition \ref{P:POLY} implies Theorem \ref{T:main_Kn2} as follows.  Since
$X_e \in [-1,1]$, the maximum total improvement for $\hH$ is at most
$n^2$. If there exists an improving sequence of length at least
$\Omega(n^{15+\eta}\phi^5)$ then there must exist an improving sequence of length
$2n$ with total improvement less than $O(n^{-(12+\eta)}\phi^{-5})$.
Apart from Section \ref{sec:wordscombinatorics} the rest of the paper is
dedicated to proving Proposition \ref{P:POLY}.

We believe that the exponent $15$ in Theorem \ref{T:main_Kn} is far from tight. In fact we make the conjecture that local max-cut is in smoothed quasi-linear time:
\begin{conjecture}
  Let $G$ be the complete graph on $n$ vertices, and assume the edge
  weights $X = (X_e)_{e\in E}$ are independent random variables with
  $|X|\leq 1$ and density bounded above by $\phi$. With
  high probability any implementation of FLIP
terminates in at most
  $n (\phi \log n)^c$ steps where $c>0$ is a universal constant.
\end{conjecture}
This quasi-linear time behavior could quite possibly extend to
an arbitrary graph $G$; however, the first step should be to show smoothed polynomial complexity in this setting
(that is, to generalize Theorem \ref{T:main_Kn} to an arbitrary graph).
Some graphs are easier than others. E.g., \cite{ET} observed that for graphs with maximum degree
$O(\log(n))$ endowed with Gaussian edge weights,  with high probability,
  any implementation of FLIP terminates in a polynomial number of steps.
(Since, with high probability, each improving move increases
$\hH$ significantly.)
In the final section of the paper we show that a  natural approach to
generalize our result to arbitrary graphs cannot work; the proof relies
on a new result on combinatorics of words, which is of independent
interest.
 
\section{Preliminaries}

In this section we provide a high-level overview of the proof of
Proposition \ref{P:POLY}.  We also state and prove some lemmas which will be useful
in our analysis.

Recall that we work in the state space $\{-1,1\}^V$ and that a move flips
the sign of a single vertex.  Each move can be viewed as a linear
operator, which we define now.  For any $\sigma \in \{-1,1\}^V$ and
$v \in V$, we denote by $\sigma^{-v}$ the state equal to $\sigma$ except
for the coordinate corresponding to $v$ which is flipped.  For such
$\sigma,v$ there exists a vector
$\alpha = \alpha(\sigma,v) \in \{-1,0,1\}^E$ such that
$\hH(\sigma^{-v}) = \hH(\sigma) + \langle \alpha, X \rangle$. More specifically $\alpha=(\alpha_{uw})_{uw \in E}$ is defined by
\begin{equation} \label{eq:alphax}
\left\{\begin{array}{ll} \alpha_{uv} = \sigma(v) \sigma(u) & \forall u \neq v \\ \alpha_{uw} = 0 & \text{if} \ v \not\in \{u,w\} \end{array}\right.
\end{equation}
Crucially, note
that $\alpha$ does not depend on $X$.  We say that $v$ is an
\emph{improving move}
from a configuration $\sigma$ if $\langle \alpha, X \rangle > 0$.  It will
be convenient to identify a  move with the corresponding vector
$\alpha$.  Thus we may talk of improving vectors (meaning that
$\langle\alpha,X\rangle>0$).  Similarly, we say that certain moves are linearly
independent if the corresponding vectors are.
%

\subsection{Basic idea for the analysis}

We first observe that for non-zero $\alpha\in\Z^E$, the random variable
$\langle \alpha, X\rangle$ also has density bounded by $\phi$.
Thus for a fixed move from $\sigma$ to $\sigma^{-v}$, one has
$\P(\langle \alpha, X \rangle \in (0,\epsilon]) \leq \phi
\epsilon$. Naively (ignoring correlations), one could expect that for a
fixed sequence of moves with corresponding vectors
$\alpha_1, \dots, \alpha_{\ell}$,
\begin{equation} \label{eq:naive}
  \P\Big(\forall i \in [\ell], \langle \alpha_i, X \rangle \in (0,\epsilon]
  \Big) \leq (\phi \epsilon)^{\ell}.
\end{equation}
A rigorous and more general statement in this direction is given in the
following lemma.

\begin{lemma}[Lemma A.1 \cite{ER}]\label{lem:proba}
  Let $\alpha_1, \dots, \alpha_k$ be $k$ linearly independent vectors in
  $\Z^E$.  Then the joint density of $\left(\langle \alpha_i, X
    \rangle\right)_{i\leq k}$ is bounded by $\phi^k$.
  In particular, if sets $J_i\subset\R$ have measure at most $\epsilon$
  each, then
  \[
    \P\Big(\forall i \in [k], \langle \alpha_i, X \rangle \in J_i\Big)
    \leq (\phi \epsilon)^k.
  \]
\end{lemma}

This lemma is stated slightly differently from Lemma~A.1 of \cite{ER} but
the same proof applies.  If in a sequence of moves all moves are
linearly independent, then \eqref{eq:naive}  holds.  Under this
assumption,   a union bound implies that the probability  there exists
an initial configuration and a sequence of $\ell$ improving moves
which improves by at most $\epsilon$ (since each step improves by at most $\epsilon$),
is smaller than $2^n n^{\ell} (\phi \epsilon)^{\ell}$, since there are
$2^n$ initial configurations and at most $n^\ell$ sequences of length
$\ell$. In other words, with high probability, any sequence of length
$\Omega(n)$  would improve the cut value by at least
$\Omega(1/\mathrm{poly}(n))$.  Since $\hH$ is bounded by
$\mathrm{poly}(n)$, as a consequence of this, the FLIP algorithm should reach
a local maximum after at most $\mathrm{poly}(n)$ steps. The challenge is to
fix the above calculation when the length of the sequence is replaced by
the linear rank of the sequence of improving moves.  A particularly
important task for us will be to show that given any sequence of length
$\Omega(n)$ of potentially improving moves, one can always find many
$\alpha_i$'s which are linearly independent.  (Some sequences of moves
cannot possibly be improving, e.g., if the same coordinate is flipped twice
in a row.)

Given an initial configuration $\sigma_0$ and a sequence of moves $L$ of
length $\ell$, let the corresponding move operators be
$\alpha_1, \dots, \alpha_{\ell}$.  Consider the $|E| \times \ell$ matrix
$\cA_L = [\alpha_i]_{i = 1}^\ell$ whose $i$th column is the vector
$\alpha_i \in \{-1, 0, 1\}^E$ (thus each row is indexed by an edge
$e \in E$).  Note that the vectors $\alpha_t$, and thus also the matrix
$\cA_L$ depends (implicitly) on the initial spin state $\sigma_0$.
The maximum number of linearly independent moves in $L$ is the
rank of the matrix $\cA_L$; and thus we may apply Lemma \ref{lem:proba} with
$k$ being this rank.

This turns out not to be sufficient for our needs.  However, if a sequence
of moves $L$ is an improving sequence from some initial state, then every
contiguous segment of $L$ is also improving from some (different) state.
We use the term \textbf{block} to refer to a contiguous segment of some
sequence of moves under consideration (we will formally define it in
Section \ref{sec:rank}).  Thus to bound the probability that $L$ is
improving we can instead consider only a segment of our choice of $L$.
Note that there are two competing effects in the choice of a segment: on
the one hand the probability that a block is $\epsilon$-slowly improving is generally much
larger than the probability that the full sequence is $\epsilon$-slowly  improving; on the
other hand any given block appears in many different sequences, which
yields an improvement in the union bound.

Our proof will proceed in two key steps: (i) find a block of $L$
with relatively high rank (this is done in Section \ref{sec:rank}),
and (ii) apply the union bound we alluded to above in a more efficient way so as
to replace the term $2^n$ (counting possible initial configurations) by a
smaller term (Section \ref{sec:union}).  To this end, we will want to find
a block in $L$ which has a high rank and in which the number of distinct
symbols is as small as possible.

\subsection{Preliminary linear algebra}

We now provide some preliminary results which prepare us to find a lower
bound for the rank of the matrix $\cA_L$ corresponding to a sequence
$L = (v_1,\dots,v_\ell)$. (Here $v_t \in V$ denotes the vertex which moves
at step $t$.)  Denote by $\sigma_t$ the spin
configuration after step $t$.
The following statement is a direct consequence of equation
\eqref{eq:alphax}.

\begin{lemma}\label{lem:alphades}
  The vector $\alpha_t$ is supported precisely on the edges incident to
  $v_t$.  The entry in $\alpha_t$ corresponding to the edge $\{v_t, u\}$ is
  $-\sigma_t(v_t) \sigma_t(u)$, which is also equal to $\sigma_{t-1}(v_t)
  \sigma_{t-1}(u)$.
\end{lemma}

We now make the following simple observation.

\begin{lemma}\label{L:invariantalpha}
  The rank of $\cA_L$ does not depend on the initial configuration
  $\sigma_0$.
\end{lemma}

\begin{proof}
  Let $\cA_L$ be obtained from some initial configuration
  $\sigma_0$ and let $\cA'_L$ be obtained from another initial configuration
  $\sigma'_0$.  Both matrices are derived from the same sequence $L$. For any vertex $u$ and time $t$ we have that
  $\sigma_t(u)\sigma'_t(u) = \sigma_0(u) \sigma'_0(u)$.
Thus the row corresponding to an edge $\{u,v\}$ in $\cA_L$
  is $\sigma_0(u)\sigma_0(v)\sigma'_0(u)\sigma'_0(v)$ times the
  corresponding row in $\cA'_L$, and thus these two matrices have the same
  rank.
\end{proof}

Rather than working with the matrix $\cA_L$ directly, we will consider the
matrix $\A=\A_L$ whose $t$-th column is $-\sigma_t(v_t) \alpha_t$ (for $t
\in [\ell]$).  Obviously $\A$ has the same rank as $\cA$; 
 in light of this and of Lemma \ref{L:invariantalpha},
we define the \textbf{rank of a sequence of moves} by
$\rank(L) = \rank(\A_L)$. For future reference, we give the
following alternative definition of the matrix $\A$ (the two definitions
are equivalent by Lemma \ref{lem:alphades}).

\begin{definition}\label{desA}
  For a given sequence $L = (v_1,\dots,v_\ell)$, let $\A=\A_L$ be the $|E|
  \times \ell$ matrix with rows indexed by edges.  For an edge $e=\{u,v\}$
  and time $t$ such that $u \neq v_t$, the entry $\A[e,t] = 1_{v_t=v} \sigma_t(u)$.
\end{definition}

Thus  the $t$-th entry of the row corresponding to an edge
$e=\{u,v\}$ is non-zero, if and only if $v_t \in \{u,v\}$.  If $v_t=v$, then
the $t$-th entry of the row $\A[\{u,v\}]$ is the spin of $u$ (the other
endpoint of the edge) at time $t$, i.e., $\sigma_t(u)$ (which   also equals
$\sigma_{t-1}(u)$ since $u \neq v = v_t$).

\section{Bounding the rank of $L$} \label{sec:rank}

The goal of this section is to prove Lemma \ref{L:rank1} which gives a lower
bound on the rank of $L$ in terms of simple combinatorial properties of
$L$. First we introduce some notation.

For any sequence of moves $L$, a vertex that appears only once in $L$ is
called a \textbf{singleton}; vertices that appear at least twice are called
\textbf{repeated vertices}.  Let $\ell(L)$ be the length of $L$; Let $s_1(L)$ be
the number of singletons in $L$, and let $s_2(L)$ be the number of repeated
vertices in $L$.  Denote by $s(L) = s_1(L)+s_2(L)$  the total number of
distinct vertices that appear in $L$.  When the sequence of moves $L$ is
clear from the context, we shall use $\ell, s, s_1$ and  $s_2$ to denote
$\ell(L), s(L), s_1(L)$ and $s_2(L)$, respectively.

A \textbf{block} of a sequence $L = (v_1, v_2, \dots, v_{\ell})$ is a
contiguous segment from the sequence, i.e.\ $(v_i, \dots, v_j)$ of length
$j-i+1$ for some $i \leq j$. We denote this block by $L[i,j]$.
A maximal block (w.r.t.\ inclusion) of $L$ which consists of only
singletons is called a \textbf{singleton block}.  A maximal block of $L$
which consists of only repeated vertices is called a \textbf{transition
  block}.
Thus $L$ is naturally partitioned into
alternating singleton and transition blocks.  Note that a repeated vertex
might appear only once in a specific transition block, in which case it
must appear also in at least one other transition block.  For every $v$ in
$L$, let $b(v)$ be the number of transition blocks containing $v$.
Let $T_1,\dots,T_k$ denote the transition blocks, and $x^+ = \max(x,0)$.
Throughout the proof, we use $u, v, w$ etc. to denote vertices in $V$;
sometimes for the purpose of enumeration, we might also use integers
$1,2,\dots$ to denote vertices in $V$ which should cause no confusion.

The next lemma is the main result of this section.

\begin{lemma}\label{L:rank1}
  For any sequence of moves $L$ one has
  \begin{enumerate}[label={(\roman*)},nosep]
  \item $\rank(L) \geq \min(s(L),n-1)$.
  \end{enumerate}
  Furthermore, if $s(L) < n$ and $L$ does not visit any state more than
  once, then
  \begin{enumerate}[label={(\roman*)},nosep,resume]
  \item $\rank(L) \geq s(L) +  s_2(L)/2.$
  \item $\rank(L) \geq s_1(L) + \sum_i s(T_i) = s(L) + \sum_{v} (b(v)-1)^+$,
    where the sum is over the transition blocks of $L$.
  \end{enumerate}
\end{lemma}

Note that $L$ visits a state more than once if $\sigma_i=\sigma_j$ for some
$i<j$, or equivalently the block $L[i+1,j]$ contains every vertex an even
number of times.  (This clearly is a property of $L$, independent of
$\sigma_0$).  If a sequence is
improving, then it cannot revisit any state.  We can safely disregard
any sequence which fails this condition in later analysis.

\begin{proof}
  (i) Without loss of generality, suppose $1, 2, \dots, s$ are the only vertices
  appearing in $L$, and suppose that $s<n$.
  Let $t_i$ be some time at which vertex $i$ appears in
  $L$; 
  Consider the $s \times s$ sub-matrix of $\A$ restricted to the columns
  ${t_i}$'s and the rows corresponding to edges $\{i,n\}$ for
  $i=1,2,\dots,s$.  By our choice of $t_i$, the column $t_i$ has a non-zero
  entry at the row corresponding to $\{i,n\}$, and no others, and thus has
  full rank $s$.  If $s=n$ apply the above reasoning to the set of times
  $\{t_1,\dots,t_{n-1}\}$.

%

  \medskip

  (ii) We first make the following simple observation.  Given a sequence
  $L$ which does not revisit any state, if vertex $v$ is moved at least
  twice, then the block between any two consecutive moves of $v$ contains
  some vertex $u$ an odd number of times in this block.  This is clear, since any
  block in $L$ contains some vertex an odd number of times by an earlier argument.

  We create an auxiliary directed graph $H$ as follows.  The vertices of
  $H$ are the $n$ vertices of $G$.  For each repeated vertex $v$, there
  must be a vertex $u$ that appears an odd number of times between the
  first two times $v$ appears. We pick one such $u$ arbitrarily, and add to
  $H$ a directed edge from $v$ to $u$.  Note that $H$ might contain both an
  edge and its reverse (e.g.\ for the sequence $L=1,2,1,3,2$).  Each
  repeated vertex has one out-going edge in $H$, and so $H$ has exactly
  $s_2$ directed edges. Moreover, directed cycles (including cycles of
  length $2$) in $H$ are vertex-disjoint, and their total length is at most
  $s_2$.  Let us define a sub-graph of $H$ by removing one edge from each
  directed cycle of $H$.  Since the cycles are vertex-disjoint (since the out-degree for each vertex is at most $1$), we remove
  at most $s_2/2$ edges, and obtain an acyclic sub-graph of $H$ with at
  least $s_2/2$ edges.

  Since not all vertices appear in $L$, suppose without loss of generality
  that vertex $n$ does not appear in $L$.  Part (ii) of the lemma now
  follows from the following.

  \begin{claim}
    For any acyclic sub-graph $H'$ of $H$, the following edges correspond
    to linearly independent rows in $\A$: All edges of $H'$, together with
    $\{v,n\}$ for vertices $v\in L$.
  \end{claim}

  We prove this by induction on the number of edges in $H'$.  If $H'$ is
  the empty subgraph, these are precisely the rows used to prove part (i).
  Now suppose $H'$ is not empty.  Since $H'$ is acyclic, there must be a
  vertex $v$ with in-degree 0 and unique outgoing edge $e = \{v,u\}$. Suppose
  we have a linear combination $\sum_i \lambda_i \A[\{i,n\}]
  +\sum_{e \in H'} \mu_{e} \A[e] = 0$, where $\A[e]$ is the row
  corresponding to $e$ and the sum is over the edges of the claim.
  Let $t_1,t_2$ be the first two times that $v$ moves. By the definition of
  $\A$ (see Definition \ref{desA}), the $t_1$-th and $t_2$-th entry of
  $\A[\{v,n\}]$ are both $\sigma_{t_1}(n) = \sigma_{t_2}(n)$ (since $n$ does not move).
  Furthermore, since $u$ appears an odd number of times between the first two
  appearance of $v$ we have that the $t_1$-th entry and $t_2$-th entry of
  $\A[e]$ are of opposite signs. Furthermore, since $v$ has out-degree 1 in
  $H'$ and in-degree 0, among the rows we have picked, only the rows
  $\A[\{v,n\}]$ and $\A[e]$ have non-zero entries in positions $t_1, t_2$.
  We thus have $\lambda_v \pm \mu_e=0$, implying $\lambda_v = \mu_e = 0$.
  Thus the linear combination involves only edges of $H'\setminus e$ and
  edges to $n$.  Applying the inductive hypothesis to $H'\setminus e$ gives
  that the linear combination is trivial.

  \medskip

  (iii) Suppose without loss of generality that $1, \dots, s_2$ are the
  repeated vertices in $L$.
  By the definition of $b(i)$, there exist times
  $t_1(i), t_2(i), \dots, t_{b(i)}(i)$ in different transition blocks at which $i$ moves, and for any
  $2 \leq j \leq b(v_i)$, there is a singleton vertex $w_{i,j}$ that
  appears in the block $L[t_{j-1}(i), t_{j}(i)]$.

  We claim that the following rows are linearly independent.  For each $v$
  in $L$ the edge $\{v,n\}$, and for each repeated vertex $i$, the rows
  $e_{i,j} = \{i, w_{i,j}\}$ for $j = 2, \dots, b(i)$.

  For any repeated vertex $v_i$, among the rows we have picked, the ones
  which
  have non-zero entries at times $t_1(i), \dots,$ $t_{b(i)}(i)$ correspond
  to the rows of $\{i,n\}$, and $e_{i,j}$ for $j = 2, \dots, b(v_i)$.
  At those columns, by Lemma \ref{L:invariantalpha}, we can assume the row $\A[\{i,n\}]$ has all ones.
  The row $\A[e_{i,j}]$ has entries $1$ before the (unique) appearance of
  $w_{i,j}$ and $-1$ after the appearance.  Thus the minor for these rows and
  the sequence of times $\{t_1(i),\dots,t_{b(i)}(i)\}$ has the form
  \[
    \begin{bmatrix}
      1 & 1 & 1 & 1 & \cdots \\
      1 &-1 &-1 &-1 & \cdots \\
      1 & 1 &-1 &-1 & \cdots \\
      1 & 1 & 1 &-1 & \cdots \\
      \vdots & \vdots & \vdots & \vdots & \ddots
    \end{bmatrix}.
  \]
  This clearly has full rank $b(i)$.
  For singleton vertices $v$ appearing at time $t=t_v$, the only selected row
  with no-zero $t$-th entry corresponds to edge $\{v,n\}$.  Thus if we
  group together columns for the repeated vertices, the selected rows of
  $\A$ have a block structure, with blocks of the form above along the diagonal
  and zeros elsewhere.  It follows that
  \[
    \rank(\A) \geq s_1 + \sum_{i\leq s_2} b(i) = s + \sum_{i\leq s_2}
    (b(i)-1)^+.  \qedhere
  \]
\end{proof}


\section{Proof of Proposition \ref{P:POLY}}
\label{sec:union}

In this section we prove Proposition \ref{P:POLY}, and thus conclude the
proof of our main result (Theorem \ref{T:main_Kn}). We first show in Subsection
\ref{sec:critical} that any improving sequence contains a certain special
block which we can use to obtain high rank.  Then we conclude the proof of
Proposition \ref{P:POLY} in Section \ref{sec:union2} with an ``improved'' union bound
argument.

\subsection{Finding a critical block with large rank}
\label{sec:critical}

We start with a simple combinatorial lemma.  Fix some $\beta>0$.  We say
that a block $B$ is \textbf{critical} if $\ell(B) \geq (1+\beta) s(B)$, and
every block $B'$ strictly contained in $B$ has $\ell(B') < (1+\beta)s(B')$.

\begin{lemma}\label{L:criticalblock}
  Fix any positive integer $n \geq 2$ and a constant $\beta > 0$.  Given a
  sequence $L$ consisting of $s(L)<n$ letters and with length
  $\ell(L) \geq (1+\beta) s$, there exists a critical block $B$ in $L$.
  Moreover, a critical block satisfies
  $\ell(B) = \lceil (1+\beta)s(B) \rceil$.
\end{lemma}

\begin{proof}
  A block satisfying $\ell(B) \geq (1+\beta)s(B)$ exists, since the whole
  sequence $L$ satisfies this.  A minimal (w.r.t.\ inclusion) block that
  satisfies this will by definition be a critical block.

  We now show that $B$ satisfies $\ell(B) = \lceil (1+\beta)s(B) \rceil$.
  If $\ell(B) \geq \lceil (1+\beta)s(B) \rceil + 1$, remove the last vertex
  from $B$, thus obtaining $B'$. Then $\ell(B') = \ell(B) -1$, while
  $s(B) \geq s(B') \geq s(B)-1$. For the block $B'$ we thus have
  \[
    \ell(B') = \ell(B) -1 \geq \lceil (1 + \beta)s(B) \rceil \geq \lceil (1
    + \beta)s(B') \rceil\, ;
  \]
  this contradicts criticality of $B$.
\end{proof}

\begin{lemma} \label{L:rank2}
  Suppose $s(B) < n$.  For a critical block $B$ as in
  Lemma \ref{L:criticalblock}, we have
  \[
    \rank(B) \geq s(B) + \frac{\beta}{1+\beta}s_1(B).
  \]
\end{lemma}

\begin{proof}
 We apply Lemma \ref{L:rank1}(iii) to $B$.  Let $T_1,\dots,T_k$ be the
  transition blocks of $B$.  If the whole of $B$ is a transition
  block, i.e.\ $s_1(B)=0$, then $s(T_1)=s(B)$ and $\rank(B)=s(B)$  by Lemma \ref{L:rank1}(iii) yields
  the claim.  Otherwise, each $T_i$ is a proper sub-block of $B$, and by
  criticality of $B$ we find $\ell(T_i) < (1+\beta) s(T_i)$ for each $T_i$.  Thus
  \begin{align*}
    \rank(B)  \geq s_1 + \sum_{\text{vertices }i \text{ in } B} b(i)   & = s_1(B) + \sum_{i=1}^k s(T_i) \\
    & \geq s_1(B) + \frac{1}{1+\beta} \sum_{i=1}^k \ell(T_i) \\
    & \geq s_1(B) + \frac{\ell(B) - s_1(B)}{1+\beta} \\
    & \geq \frac{\ell(B)}{1+\beta} +  \frac{\beta}{1+\beta}s_1(B),
  \end{align*}
  where we have used that $\ell(B) = s_1(B) + \sum_{i=1}^k \ell(T_i)$, since each
  letter is either a singleton or part of one of the $T_i$.

  By Lemma \ref{L:criticalblock}, $\ell(B) = \lceil(1+\beta)\rceil s(B)$,
  and the claim follows.
\end{proof}

\begin{corollary}\label{criticalrank}
  For a critical block $B$ with $s(B)<n$, we have
  \[
    \rank(B) \geq s(B) + \max\left(\frac{\beta}{1+\beta}s_1(B), \frac12
      s_2(B)\right).
  \]
  In particular,
  \[
    \rank(B) \geq \frac{1+4\beta}{1+3\beta} s(B).
  \]
\end{corollary}

\begin{proof}
  The two bounds come from Lemmas \ref{L:rank1} and \ref{L:rank2}.  Since
  $s_1(B) + s_2(B) = s(B)$, the last bound is obtained by a convex
  combination of the two preceding bounds.
\end{proof}

\subsection{A better bound on improving sequences}
\label{sec:better}

Lemma \ref{lem:proba} implies that the
probability that a sequence $L$ is $\epsilon$-slowly improving from any given
$\sigma_0$ is at most $(\phi\epsilon)^{\rank(L)}$, and therefore the
probability that $L$ is $\epsilon$-slowly improving from \emph{some}
$\sigma_0$ is at most $2^n (\phi\epsilon)^{\rank(L)}$.  For sequences with
large rank this is sufficiently small for our needs.  However, for sequences
with small rank and small $s$ a better
bound is needed.  The next novel ingredient of our proof is an improvement
of this bound that reduces the factor of $2^n$, provided $s(L)$ is small.

\begin{lemma}\label{L:betterprob}
  Suppose the random weights $X_e$ a.s.\ have $|X_e|\leq 1$.  Then
  \[
    \P(\text{$L$ is $\epsilon$-slowly improving from some $\sigma$}) \leq
    2 \left(\frac{4n}{\epsilon}\right)^s (8\phi\epsilon)^{\rank(L)}.
  \]
\end{lemma}

The key idea is that instead of taking a union over the initial state
$\sigma_0$ for the non-moving vertices, we only consider the influence of the non-moving vertices on
the moving vertices.

\begin{proof}
  Without loss of generality, we may assume that the vertices that appear
  in $L=L[1,\ell]$ are $1,\dots,s$, and that $s+1,\dots,n$ do not appear.  We
  separate $\hH(\sigma) = \hH_0(\sigma) + \hH_1(\sigma) + \hH_2(\sigma)$,
  where $\hH_j$ is the sum over edges with $j$ endpoints that appear in
  $L\cup\{s+1\}$, for $j\in\{0,1,2\}$.  The reason for including $s+1$
  will become clear later.

  With a given initial state $\sigma_0$, let $\sigma_t$ be the state after
  flipping
  the state of $v_t$.  For $u>s$ (so $u$ does not appear in $L$), we have that
  $\sigma_t(u)$ is constant over $t \le \ell$ and thus
  $\hH_0(\sigma_t) = \hH_0(\sigma_0)$ for all $t \le \ell$.  Moreover, as in
  \eqref{eq:alphax}, we get
  \[
    \hH_1(\sigma_t) - \hH_1(\sigma_{t-1}) =
    - \sigma_t(v_t) \sum_{u=s+2}^n X_{v_t,u} \sigma_0(u) =
    \sigma_t(v_t) Q(v_t),
  \]
  where $Q(v) = - \sum_{u=s+2}^n X_{v_t,u} \sigma_0(u)$.  One may think
  of $Q$ as a constant \emph{external field} acting on the $s$ moving
  vertices.  Finally, the increments of $\hH_2$ are linear functionals of
  the weights on edges with both endpoints in $\{1,\dots,s,s+1\}$.  We denote
  these functionals by $\bar\alpha_t$, so that
  \[
    \hH(\sigma_t)-\hH(\sigma_{t-1}) =
    \sigma_t(v_t) Q(v_t) + \langle \bar\alpha_t, X\rangle.
  \]
  Note that $\bar\alpha_t$ is simply the restriction of $\alpha_t$ to edges
  with both endpoints in $\{1,\dots,s+1\}$.  Observe that $\bar\alpha_t$
  depends on the first $s+1$ coordinates of $\sigma_0$, but not on the
  other coordinates.

  Since $X_e$ is assumed to be bounded, we have $|Q(v)|\leq n$.  Consider
  the set $D = 2\eps\Z \cap [-n,n]$, of size at most $n/\eps+1\leq
  2n/\eps$.  We have that $Q(v)$ is within $\eps$ of some element of
  $d(v) \in D$.  Instead of a union bound on $\sigma_0$, we now use a union
  bound over $(\sigma_0(i))_{i\leq s+1}$ and the vector $(d(v))_{v\leq s}$.
  From the above definitions it follows that
  \[
    \hH(\sigma_t) - \hH(\sigma_{t-1}) =
    \langle \bar\alpha_t, X\rangle + \sigma_t(v_t)d(v_t) + \delta_t,
  \]
  where $|\delta_t|\leq \eps$.
  If the sequence is $\eps$-slowly increasing, then
  \[
    \Big|\langle \bar\alpha_t, X\rangle  + \sigma_t(v_t)d(v_t) \Big| \leq 2\eps,
  \]
  and thus $\langle\bar \alpha_t, X\rangle$ lies in the union of two
  intervals of length $4\eps$ centered at $\pm d(v_t)$.  Note that
  $\rank(\bar\alpha_t) = \rank(L)$, since we included in $\bar\alpha$
  the contributions from the stationary vertex $s+1$. (This holds also
  if $s=n$.)  
  By Lemma \ref{lem:proba}, the probability of this event is at most $(8\eps
  \phi)^{\rank(L)}$.  Crucially, if we know $(\sigma_0(i))_{i\leq s+1}$ and
  $d(v)$ for $v=1,\dots,s$, then the event under consideration is the same
  for all $2^{n-(s+1)}$ possible configurations $\sigma_0$.

  The claim now follows by a union bound over the possible values of
  $(\sigma_0(i))_{i\leq s+1}$ and $d(v)$.
\end{proof}


\subsection{Proof of Proposition \ref{P:POLY}}
\label{sec:union2}

\begin{proof}[Proof of Proposition \ref{P:POLY}]
  Fix $\beta = 1$.  Let $R$ be the event that
  there exists an initial configuration $\sigma_0$ and a sequence $L$ of
  length $2n$ which is $\eps$-slowly improving.  Our goal is to show $\P(R) =
  o(1)$.

  We consider two cases: either the sequence $L$ has $s(L) = n$ or else
  $s(L) < n$.  Call these events $R_0$ and $R_1$.
  We bound $\P(R_0)$ by a union bound over sequences:
  \begin{equation}
    \label{eq:all_move}
    \P(R_0) \leq  \sum_{\sigma_0} \sum_{L : s(L) = n} \P(\text{$L$ is
      $\eps$-slowly improving from $\sigma_0$}).
  \end{equation}
  The summation is over all initial configurations $\sigma_0$ and all
  possible sequences of improving moves $L$ from $\sigma_0$ with $n$ moving
  vertices.  There are $2^n$ initial configurations and at most $n^{2n}$
  sequences of length $2n$.  Since $s = n$, each such sequence has
  $\rank(L) \geq n-1$ by Lemma \ref{L:rank1}(i).  By Lemma \ref{lem:proba}, each term
  in \eqref{eq:all_move} is bounded by $(\phi\eps)^{n-1}$, and so
  \begin{equation}
    \label{eq:PR0}
    \P(R_0) \leq 2^n n^{2n} (\phi\eps)^{n-1} = o(1),
  \end{equation}
  provided $2n^2\phi\eps$ is small.

  \medskip

  We turn to the event $R_1$, that there exists an initial configuration
  $\sigma_0$ and an $\eps$-slowly improving sequence $L$ of length $2n$
  such that $s(L) < n$.  By Lemma \ref{L:criticalblock}, on the event $R_1$ for
  some $s<n$ there exists a critical block using precisely $s$ vertices and
  some initial configuration such that the block is $\eps$-slowly improving
  from that configuration.  Thus
  \begin{equation}
    \label{eq:some_move}
    \P(R_1) \leq \sum_{\text{critical $B$}} \P(\text{$B$ is
      $\eps$-slowly improving from some $\sigma$}).
  \end{equation}

  By definition, a critical block has $\ell(B)=2s(B)$.   By
  Corollary \ref{criticalrank}, it has $\rank(B) \geq 5s(B)/4$.  Thus by
  Lemma \ref{L:betterprob}, for any critical block we have
  \begin{align*}
 & \P(\text{$B$ is $\eps$-slowly improving from some $\sigma$}) \\
 \leq &
  2 \left(\frac{4n}{\eps}\right)^{s(B)} (8\phi\eps)^{5s(B)/4}
  \leq 2\left(64\phi^{5/4} n \eps^{1/4} \right)^{s(B)}.
  \end{align*}

  The number of critical blocks using $s$ letters is at most $n^{2s}$,
  (which is the number of sequences of length $\ell=2s$).  Thus
  \begin{equation}
    \label{eq:some_move1}
  \P(R_1) \leq 2\sum_{s<n} n^{2s} \left(64\phi^{5/4} n \eps^{1/4}
  \right)^s.
  \end{equation}
  This sum tends to $0$ as $n\to\infty$ when $\eps =
  n^{-(12+\eta)}\phi^{-5}$ with $\eta>0$.
\end{proof}

\begin{remark}
  The proof above shows that for $\eps = \alpha \phi^{-5} n^{-12}$, we
  have $\P(R_1) \leq O(\alpha^{3/4})$ as $\alpha\to 0$ (since a critical
  block with $\beta=1$ has $s\geq 3$) and hence that the run time of the
  FLIP algorithm, divided by $n^{15}$ is tight.  The number of critical
  blocks with a given $s$ can be bounded by
  $\binom{n}{s} s^{2s} \leq (ens)^s$ which is less than $n^{2s}$ for
  $s\leq n/e$. Using this gives 
  \begin{equation}
    \label{eq:some_move1imp}
    \P(R_1) \leq 2\sum_{s<n} \left(C \phi^{5/4} n^2 s \eps^{1/4} \right)^s,
  \end{equation}
  and so $\P(R_1)$ decays super-polynomially in $\alpha$.
\end{remark}

Corollary~\ref{C:main_Kn} follows easily from the proof of
Proposition~\ref{P:POLY}:

\begin{proof}[Proof of Corollary~\ref{C:main_Kn}]
  Suppose an increasing sequence of length $L\geq 2n$ exists.  Since the
  total weight of any cut is in $[-n^2/4,n^2/4$, there must be a block of
  size $2n$ in $L$ such that the total improvement along the block is at
  most $\epsilon = \frac{n^2}{2[L/2n]} \leq 2n^3/L$.  Let
  $R(n,L)$ be the probability there is such a block using all $n$
  letters ($R_0$ above), and $R(s,L)$ the probability there is a
  critical block of length $2s$ using $s$ letters.

  Let $T$ be the number of steps before FLIP terminates.  Then we have
  \[
    \P(T \geq L) \leq \sum_{s\leq n} \P(R(s,L)).
  \]
  and so
  \[
    \mathbb{E}(T) =  \sum_{L=1}^\infty \P(T \geq L)
    \leq n^{15} + \sum_{L>n^{15}} \sum_{s\leq n} \P(R(s,L)),
  \]
  and we need to show that the last sum is $O(n^{15})$.
  For $s=n$, by \eqref{eq:PR0},
  \[
    \P(R(n,L)) \leq 2^n n^{2n} (\phi\eps)^{n-1} = 2^n n^{2n} (\phi 2n^3
    / L)^{n-1},
  \]
  and the sum over $L>n^{15}$ is $o(n^{15})$.

  For $s>4$, by \eqref{eq:some_move1imp},
  \[
    \P(R(s,L)) \leq 2 \left(C \phi^{5/4} n^2 s (2n^3/L)^{1/4}
    \right)^s,
  \]
  and so
  \[
    \sum_{L>n^{15}} \P(R(s,L)) \leq (Cs)^s n^{11s/4} \frac{C}{s}(n^{15})^{1-s/4}
    \leq (Cs)^{s-1} n^{15-s}.
  \]
  Thus $\sum_{4<s<n} \sum_{L>n^{15}} \P(R(s,L)) = O(n^{15})$.
  
  For small $s$ the bound above is not sufficient, and we need a better
  rank bound.  There are no critical blocks with $s=1$ or $s=2$. It is
  easy to check that critical blocks with $s=3$ all have rank $6$.  A
  short exhaustive search yields that critical blocks with $s=4$ have
  rank $7$ or $8$.  Since the number of sequences with $s=3$ or $s=4$ is
  $O(n^s)$, for $s=3,4$ we get
  \[
    \P(R(s,L)) \leq O\left( n^s \left(\frac{4n}{\epsilon} \right)^s
      (8\phi\epsilon)^{s+3}\right) = O(n^{2s}\epsilon^3).
  \]
  Thus $\sum_{s=3,4} \sum_{L>n^{15}} \P(R(s,L)) = o(1)$, which completes
  the proof.
\end{proof}

\section{A word that is sparse at every scale} \label{sec:wordscombinatorics}

The quasi-polynomial proof in \cite{ER} (which applies to any graph) relied crucially on the following lemma:
for any word of length $\ell = \Omega(n)$ over an
alphabet of size $n$, there must exist a subword of some length $\ell'$
such that the number of distinct letters which appear more than once in
this subword is $\Omega(\ell' / \log(n))$ (see Lemma \ref{lem:combi} below for a precise statement).
In some sense this says that ``a
word cannot be too sparse at every scale'' (a word is viewed as sparse if
it is mostly made of letters that appear only once).
We provide here a simple new proof of this
statement.
A natural approach to prove smoothed polynomial complexity for any graph (that is generalize
Theorem \ref{T:main_Kn} to arbitrary graphs) would be to remove the $\log(n)$ term in this combinatorics of words lemma
(see paragraph after Lemma \ref{lem:combi} for more details).
Our main contribution is this section is to show that such an improvement is not possible: we show by a probabilistic
construction that Lemma \ref{lem:combi} is tight, that is there exist words which are
sparse at every scale to the extent allowed by the lemma.  More
specifically we construct a word of length $\Omega(n)$ such that for any
subword of length $\ell'$ the number of repeating letters is $O(\ell' /
\log(n))$ (in fact we prove a stronger version of this statement where
$\ell'$ is replaced by the number of distinct letters in the subword), see
Theorem \ref{T:sparse} below.


\begin{lemma}\label{lem:combi}
  Suppose $a>1$, and that $L$ is a sequence of length $an$ in an alphabet
  of $n$ letters.  Then there exists a block $B$ in $L$ such that
  \[
  \frac{s_2(B)}{s(B)} \geq \frac{s_2(B)}{\ell(B)} \geq \frac{a-1}{a\log_2(n)}.
  \]
\end{lemma}

\begin{proof}
  The first inequality holds trivially for every block $B$.
  Define the surplus of a sequence $L$ to be $\ell(L)-s(L)$, i.e.\ the
  difference between the number of elements and the number of distinct
  elements in the sequence.  If a block $B$ is a concatenation of $B_1$ and
  $B_2$ then its surplus is at most the total surplus of $B_1$ and $B_2$
  plus $s_2(B)$.

  Let $m(\ell)$ be the maximum surplus in any block of length $\ell$ in
  $L$.  Assume that for some $\eps$, for every block $B$ from $L$ we have
  $s_2(B) \leq \eps \ell(B)$.  Then one has
  \[
  m(2\ell) \leq 2 m(\ell) + \eps\cdot2\ell.
  \]
  By recursing this inequality, with $m(2\ell-1)\leq m(2\ell)$ and $m(1)=0$
  we get
  \[
  m(an) \leq \eps an \log_2(n).
  \]
  Since $m(an) = an-s(L) \geq (a-1)n$, this shows that $\eps$ has to be
  greater than $\frac{a-1}{a\log_2(an)}$ which concludes the proof.
\end{proof}

It is easy to check that the proof of the rank lower bound given in
Lemma \ref{L:rank1}(ii) (and (i)) applies to arbitrary graphs.  By using
Lemma \ref{lem:combi} above together with the union bound argument from
Section \ref{sec:union2} one obtains an alternative proof to the quasi-polynomial
complexity result of \cite{ER}.  A tempting approach to prove a polynomial
complexity result for any graph would be to ``simply'' replace the
$\log(n)$ term in Lemma \ref{lem:combi} by some constant.  The main result of
this section is to show that this cannot be done, and that the $\log(n)$ in
Lemma \ref{lem:combi} is tight up to possibly constant factors.  As noted above,
this can be interpreted as saying that there exist
words which are sparse at every scale.
In fact, we prove something stronger, as stated in the following theorem.

\begin{theorem}\label{T:sparse}
  For every $a>1$ there exists a $C$ so that for every $n$ there is a
  sequence of length $\ell = [an]$ in $n$ letters so that every block $B$
  of $L$ has $s_2(B)/s(B) \leq C/\log n$.  Moreover, for $n>n_0(a)$ one may
  take $C=9a\log(a)$.
\end{theorem}

This is stronger in that we have a bound on $s_2(B)/s(B) \geq
s_2(B)/\ell(B)$.  We remark that decreasing $a$ makes the problem easier
(just take the first $[a'n]$ letters).  We can assume all letters are used
in the sequence, otherwise we can replace some repetitions by unused letters.

\subsection{The probabilistic construction}

The construction proving Theorem \ref{T:sparse} is probabilistic, and implies that
there are many sequences with these properties. We do not optimize the
constant $C$ here in order to keep the proof simple and clean.  A more
careful analysis will improve $C$.

We create a sequence as follows.  In \textbf{stage one} of the construction
we write
down the (potentially) repeated letters.  Each repeated letter is written
in some random set of locations, possibly overwriting previous letters.
Afterwards, in \textbf{stage two}, all positions where no repeated letters
have been
written are filled in with new and unique letters.  Note that it is
possible that a potentially repeated letter is overwritten, and
consequently appears only once or even not at all in the final sequence.

The construction is defined in terms of integers $b_0,b_1$ and $\gamma$ which we will specify later in the proof.
The potentially repeated letters are denoted by $i$ and $i'$ for
$i\in\{b_0,\dots,b_1-1\}$. Thus the total number of potentially repeated
letters is $2(b_1-b_0)$. To simplify the description, we construct an
infinite sequence and truncate afterwards to the first $\ell$ letters. For
each $i \in [b_0,b_1)$, split $\N$ to blocks of size $\gamma i$. In each
block $[k\gamma i,(k+1)\gamma i)$ where $k \in\N$, we choose uniformly one
position; In that position write the letter $i$ if $k$ is even, and $i'$ if
$k$ is odd.  All these choices are independent.  (Creating an infinite
sequence at this stage avoids having shorter blocks at the end.)  A
position that is left empty at the end of stage one is filled in stage two.

\subsection{Negative correlations}

For $t\leq\ell$, let $U_t$ be the event that position $t$ is empty at the
end of stage one.  We will prove that any block contains many unique
letters.  If the $U_t$ were independent this would follow from standard
large deviation bounds for Binomial random variables.  While the $U_t$ are
not independent, they have a weaker property which is sufficient for our
needs.  A collection of events $\{U_t\}$ is called \textbf{negatively
  correlated} if for every subset $S$ of indices and every $t\not\in S$ we have
\begin{align}
  \P(U_t | U_s \, \forall s\in S) &\leq \P(U_t), \label{eq:nc1}\\
  \P(U_t^c | U_s^c \, \forall s\in S) &\leq \P(U_t^c). \label{eq:nc2}
\end{align}
Negative correlation of the $(U_t)$ will follow from the following more
general statement.

\begin{proposition}
  Let $A_1,\dots,A_m$ be some finite sets, and pick a uniform element from
  each set independently.  Let $U_x$ be the event that element $x$ is never
  picked.  Then the $U_x$ are negatively correlated.
\end{proposition}

This applies to our model, by taking the sets to be the intervals $[k\gamma
i,(k+1)\gamma i)$ for $b_0\leq i< b_1$ and all $k$.

\begin{proof}
  The effect of conditioning on $U_s \, \forall s\in S$ is simple: The
  element from $A_i$ is chosen uniformly from $A_i \setminus S$.  Clearly
  this can only decrease the probability that an element $t$ is not
  selected from any $A_i$.  Since selections are independent, this gives
  \eqref{eq:nc1}.

  Now we prove \eqref{eq:nc2}.  The claim is equivalent to proving
  \[
  \P(U_t | U_s^c\, \forall s\in S) \geq \P(U_t),
  \]
  which in turn is equivalent to
  \[
  \P(U_s^c\, \forall s\in S | U_t) \geq \P(U_s^c\, \forall s\in S).
  \]
  Let $a_i$ be the element picked from $A_i$.
  To obtain the law of $(a_i)$ conditioned on $U_t$, start with the
  unconditioned selections, and resample each $a_i$ if $a_i=t$, until
  another element is chosen.  If initially (in the unconditioned vector),
  every element of $S$ is selected from some $A_i$, then this is also true
  after the resampling, and so the probability of such full occupation is
  increased.
\end{proof}

We use the following generalized Chernoff bounds for negatively correlated
events. 

\begin{theorem}[\cite{PS}]\label{T:NC_Chernoff}
  Suppose $U_1, \dots, U_k$ are negatively correlated events, and let
  $Y = \sum_{i=1}^k 1_{U_i}$ be the number of bad events occur.  Then for any
  constant $\delta \in (0,1)$,
  \[
  \P(Y\leq (1-\delta) \E[Y]) \leq \left( (1-\delta)^{-(1-\delta)}
  e^{-\delta} \right)^{\E[Y]} \label{lowertail} \; .
  \]
\end{theorem}

\subsection{Analysis of the construction}

We first estimate the probability that a letter of the sequence is filled
in stage two.  This probability is $\P(U_t) = \prod_{i=b_0}^{b_1-1} \left( 1
- \frac{1}{\gamma i} \right)$, which we denote by $d$.

\begin{lemma}
  \[
  \left(\frac{b_0-1}{b_1-1}\right)^{1/\gamma} \leq d \leq
  \left(\frac{b_0}{b_1}\right)^{1/\gamma}.
  \]
\end{lemma}

\begin{proof}
  Let $f(x) = \prod_{i=b_0}^{b_1-1} \left( \frac{i - 1/\gamma + x}{i+x}
  \right)$.  Then $f$ is increasing in $x$, and $d=f(0)$.
  We have that
  \[
  d^\gamma \leq f(1/\gamma) \cdot f(2/\gamma) \cdots f(1) = \frac{b_0}{b_1},
  \]
  as this is a telescoping product.  Similarly,
  \[
  d^\gamma \geq f(0) \cdot f(-1/\gamma) \cdots f((1-\gamma)/\gamma) =
  \frac{b_0-1}{b_1-1}.  \qedhere
  \]
\end{proof}

\begin{proof}[Proof of Theorem \ref{T:sparse}]
  With $a>1$ and $n$ given, we apply the probabilistic construction above
  with parameters
  \begin{align*}
    b_0 &= [\log n]
    & b_1 &= \left[ \sqrt{n} \right]
    & \gamma = \left[ \frac{\log n}{2\log(2a)} \right].
  \end{align*}
  Note that $b_0/b_1 = n^{-1/2+o(1)}$, and therefore $d$ tends to
  $\frac{1}{2a}$ as $n\to\infty$.

  We first claim that with good probability the resulting sequence uses at
  most $n$ letters. Stage one uses at most $2b_1 = 2\sqrt{n}$ letters. The
  expected number of letters used in stage two is $d\ell \leq n/2 + o(n)$.
  By Markov's inequality, the whole sequence use at most $n$ letters with
  asymptotic probability at least $1/2$.

  \medskip

  Next we consider repetitions within (possibly smaller) blocks. Since
  occurences of the letter $i$ are at least $\gamma i$ apart, and similarly
  for the letter $i'$, not all letters can appear multiple times in short
  blocks. In particular, each block $B\in L$ is certain to have $s_2(B)
  \leq 2\ell(B)/\gamma$. Moreover, blocks $B$ with $\ell(B)\leq \gamma b_0$
  have no repeated letters by our construction, so that $s_2(B)=0$ for such blocks.

  To estimate $s(B)$, we note that the number of letters in $B$ is at least
  the number of letters added to $B$ in stage two:
  \[
  s(B) \geq u(B) := \sum_{t=i}^j 1_{U_t}.
  \]
  We have $\E u(B) = d \ell(B)$. By the Chernoff bound Theorem \ref{T:NC_Chernoff}
  with $\delta=1/2$ we have
  \[
  \P\left( u(B) \leq \tfrac12 d \ell(B) \right) \leq
  (\sqrt{2/e})^{d\ell(B)}.
  \]
  For blocks of length at least $\gamma b_0$ this is $e^{-c\log^2 n} =
  o(n^{-2})$.  By a union bound, with high probability every block of
  length at least $\gamma b_0$ has
  \[
  s_2(B) \leq 2 \ell(B)/\gamma
  \qquad \text{and} \qquad
  s(B) \geq \frac{d \ell(B)}{2},
  \]
  and so $\frac{s_2(B)}{s(B)} \leq \frac{4}{d\gamma}$. (Shorter blocks have
  $s_2(B)=0$.)

  As $n\to\infty$, this decays as $\frac{8a\log(2a)+o(1)}{\log(n)}$,
  implying the claim for $n$ large enough. By changing $C$ we can get the
  claim also for all smaller $n$.
\end{proof}

\begin{remark}
  The above construction can be used to show that for any $a>0$ and
  $\eta>0$ there exist infinitely many graphs $G$ (with number of vertices
  tending to infinity), paired with some initial configurations $\sigma_0$ and
  sequence of moves $L$, such that $\ell(L) \geq a |V(G)|$, and for each
  block $B \in L$, $\rank(B) \leq (1+\eta) s(B)$.  These graphs are a
  significant obstacle to generalizing our main result (Theorem \ref{T:main_Kn})
  beyond the complete graph via rank arguments.
\end{remark}

\begin{acks}

We are grateful to Constantinos Daskalakis for bringing this problem to our attention, and for helpful discussions at an early stage of this project. We thank the Bellairs Institute, where this work was initiated. Most of this work was done at Microsoft Research Redmond during the first author's visit and the last author's internship.
O. Angel is supported in part by NSERC.

\end{acks}

\bibliographystyle{ACM-Reference-Format}
\bibliography{bib}

\end{document}